\documentclass[%
 reprint,
 amsmath,amssymb,
 aps,
pra,
]{revtex4-2}

\usepackage{graphicx}% Include figure files
\usepackage{dcolumn}% Align table columns on decimal point
\usepackage{amsmath,amssymb,amsfonts,amsthm,mathtools,bm}% bold math
\usepackage{caption}
\usepackage{subcaption}
\usepackage{xcolor}

\newtheorem{proposition}{Proposition}
\theoremstyle{definition}
\newtheorem{definition}{Definition}

\DeclareMathOperator{\Trc}{Tr}

\newcommand{\argmax}{\operatorname*{argmax}}

\newcommand{\Wiso}{d_{\mathrm{iso}}}
\newcommand{\Wtheta}{d_{\theta}}
\newcommand{\dM}{d_{\mathcal M}}
\newcommand{\etal}{\eta}
\newcommand{\nth}{n_{\mathrm{th}}}
\newcommand{\Ntot}{N_{\mathrm{tot}}}

\newcommand{\Ndisp}{N_{\mathrm{disp}}}

\newcommand{\out}{\mathrm{out}}
\newcommand{\refstate}{\mathrm{ref}}

\begin{document}

\title{From State-Space Transport to Measurement-Aware Distinguishability \\ 
in Quantum Sensing}% Force line breaks with \\

\author{Arnaud Coatanhay}
\email{arnaud.coatanhay@ensta.fr}
\author{Angélique  Drémeau}%
\email{angelique.dremeau@ensta.fr}
\affiliation{%
 Lab-STICC, UMR CNRS 6285, ENSTA, Institut Polytechnique de Paris,\\
2 rue Fran\c{c}ois Verny, 29806 Brest Cedex 9, France
}%

\date{\today}% It is always \today, today,

\begin{abstract}

We introduce a hierarchy of Wasserstein-type finite-distance contrasts for single-mode Gaussian sensing through a thermal-loss channel. The framework is organized by the physical level of the observation: a global Gaussian moment-space metric $d_{\mathrm{iso}}$, a homodyne-projected pseudometric $\Wtheta$, and a measurement-aware distance $\dM$ evaluating noisy detector-output distributions. At fixed input energy, $d_{\mathrm{iso}}$ locally decreases when energy is transferred from coherent displacement to squeezing, whereas the projected and measurement-aware contrasts capture a measurement-dependent competition between displacement and covariance deformation. Because projection and common additive detector noise are contractive, this hierarchy obeys the ordering $\dM\le \Wtheta\le d_{\mathrm{iso}}$. We compare these finite-distance objectives with classical and quantum Fisher information for transmissivity estimation, demonstrating that their respective optimal resource allocations generally diverge. Finally, we apply this framework to a block-fading environment, utilizing $\dM$ as an empirical finite-sample channel-monitoring statistic. Benchmarking its calibrated transmissivity estimate against a mean-only estimator and the exact Gaussian maximum-likelihood estimator clarifies its practical operational role. Ultimately, this approach establishes task-specific transport geometries without conflating them with asymptotic discrimination errors or local metrological sensitivity.

\end{abstract}

\maketitle

%\tableofcontents

% =================================================
\section{Introduction}
% =================================================

Loss and noise alter a quantum probe before the measurement record is formed, and these two stages need not be described by the same notion of distinguishability. In quantum sensing, overlap quantities such as fidelity and the quantum Chernoff bound are tied to asymptotic state discrimination, whereas classical and quantum Fisher information (CFI and QFI) quantify local sensitivity to parameter variations. A distinct operational question is whether the physical outputs associated with two finite parameter values remain structurally far apart with respect to a prescribed geometry of states or measurement outcomes. Wasserstein transport provides precisely such a finite-distance notion because it penalizes not only probability overlap, but also the physical displacement and variance deformation of the probability mass \cite{Villani2009,Gelbrich1990,Bhatia2019,DePalmaTrevisan2021}.

This distinction becomes critical when a sensing architecture contains several physical layers. Before measurement, one may compare the complete Gaussian first- and second-moment data. After choosing a homodyne phase, only one marginal distribution is observed. Detector noise then further modifies that classical distribution. These successive physical maps inherently discard information, so a transport quantity defined before the measurement and one defined after the detector describe fundamentally different objects. The purpose of this work is to make this hierarchy explicit and analytically tractable in a thermal-loss model.

To ensure conceptual clarity, we employ precise geometric terminology throughout. A \emph{metric} is a nonnegative, symmetric distance satisfying the triangle inequality, with $d(x,y)=0$ strictly implying $x=y$. A \emph{pseudometric} obeys the same properties except that distinct points may have zero distance. Because a fixed homodyne measurement can map different quantum states to identical classical distributions, the Wasserstein distance between homodyne-output laws is a true metric on those classical laws, but acts only as a pseudometric when pulled back to the quantum-state space. Unless stated otherwise, $W_2$ relies on the squared Euclidean ground cost on the specified classical coordinate space. Furthermore, we use the term \emph{finite-distance contrast} to denote the numerical value of these transport distances between a received output and a specified reference at a finite parameter separation. This term carries no implicit claim of optimal discrimination or estimation performance; whenever such a specific operational task is intended, we state it explicitly and employ the corresponding discrimination or Fisher-information criterion.

Our model relies on a displaced squeezed Gaussian probe with a fixed total photon-number budget, transmitted through a single-mode thermal-loss channel. We first treat the intensity transmissivity $\eta$ as a fixed channel parameter. Within this static framework, we derive three levels of transport comparison: (i) an isotropic Gaussian moment-space distance $d_{\mathrm{iso}}$, (ii) a homodyne-projected distance $\Wtheta$, and (iii) a measurement-aware distance $\dM$ incorporating additive detector noise. We analytically prove the local non-optimality of squeezing for the isotropic objective, demonstrate that the homodyne angle optimization reduces to a boundary choice for the projected and noisy objectives, and establish the contraction relation $\dM\le \Wtheta\le d_{\mathrm{iso}}$.

We then compare these finite-distance objectives with CFI and QFI for the estimation of the same transmissivity. The different criteria generally select different optimal resource allocations between displacement and squeezing. This divergence is not a discrepancy: while CFI and QFI quantify local parameter-estimation sensitivity, the transport quantities considered here evaluate the macroscopic geometric distance to a fixed reference background.

Finally, we translate this thermal-loss model into a practical application by allowing $\eta$ to vary from one measurement block to another. We define this explicitly as a \emph{block-fading channel}: $\eta$ remains constant during a block of $M$ homodyne samples and is redrawn independently between blocks. The measurement-aware transport contrast can then be directly estimated from the finite empirical record and calibrated to produce a scalar transmissivity monitor. Benchmarking this estimator against a mean-only estimator and the full Gaussian maximum-likelihood estimator (MLE) gives the transport statistic a concrete, yet deliberately limited, operational role.

The paper is organized as follows. Section \ref{sec:model} defines the thermal-loss sensing model and the comparison criteria used throughout. Section \ref{sec:isotropic} develops the isotropic Gaussian transport distance and its energy-allocation results. Section \ref{sec:measurement_hierarchy} passes from state-level moments to homodyne and noisy detector-output distributions. Section \ref{sec:fisher} positions the finite-distance objectives relative to CFI and QFI. Section \ref{sec:fading} applies the measurement-aware construction to finite-sample block-fading monitoring. Finally, Section \ref{sec:frequency} discusses the single-mode limitation and outlines its frequency-dependent extensions.

% =================================================
\section{Sensing model and comparison criteria}
\label{sec:model}
% =================================================

\subsection{Thermal-loss channel and probe family}

We use quadratures satisfying $[x,p]=i$, so that the vacuum covariance is $V_{\mathrm{vac}}=I_2/2$. An input single-mode Gaussian state is specified by its mean vector $\nu_{\mathrm{in}}$ and covariance matrix $V_{\mathrm{in}}$. A thermal-loss channel with \emph{intensity transmissivity} $\eta\in[0,1]$ and mean environmental photon number $\nth\ge0$ maps these moments to
\begin{equation}
\nu_{\mathrm{out}}=\sqrt{\eta}\,\nu_{\mathrm{in}},\qquad
V_{\mathrm{out}}=\eta V_{\mathrm{in}}+(1-\eta)V_{\mathrm{th}},
\label{eq:channel_outputs}
\end{equation}
where
\begin{equation}
V_{\mathrm{th}}=n I_2,\qquad n:=\nth+\frac12.
\label{eq:thermal_covariance}
\end{equation}
The corresponding centered thermal state is denoted $\rho_{\mathrm{ref}}$. It is exactly the channel output at $\eta=0$ and serves as the fixed reference background throughout the paper:
\begin{equation}
\nu_{\mathrm{ref}}=0,\qquad V_{\mathrm{ref}}=n I_2.
\label{eq:reference_state_moments}
\end{equation}
Until Sec.~\ref{sec:fading}, $\eta$ is treated as a deterministic parameter. Section \ref{sec:fading} will explicitly replace it with a block-wise random variable to model a fading environment, without altering the underlying conditional channel model of Eq.~(\ref{eq:channel_outputs}).

The probe has a fixed total mean photon-number budget $\Ntot$, which is divided between coherent displacement and squeezing. We write
\begin{equation}
\Ndisp=(1-\beta)\Ntot,\qquad
N_{\mathrm{sq}}=\beta \Ntot,\qquad \beta\in[0,1],
\label{eq:energy_split}
\end{equation}
with the real coherent amplitude $\alpha=\sqrt{\Ndisp}$ and the squeezing parameter $r\ge0$ defined by $N_{\mathrm{sq}}=\sinh^2 r$. We choose the displacement and squeezing axes to be $x$, yielding
\begin{equation}
\nu_{\mathrm{in}}=
\begin{pmatrix}\sqrt{2}\,\alpha\\0\end{pmatrix},\qquad
V_{\mathrm{in}}=\frac12
\begin{pmatrix}e^{-2r}&0\\0&e^{2r}\end{pmatrix}.
\label{eq:input_moments_squeezed}
\end{equation}
For $r>0$, $x$ (corresponding to projection angle $\theta=0$) is therefore the squeezed quadrature, and $p$ (corresponding to $\theta=\pi/2$) is the anti-squeezed quadrature. The received moments are
\begin{equation}
\nu_{\mathrm{out}}=
\begin{pmatrix}\sqrt{2\eta \Ndisp}\\0\end{pmatrix},\qquad
V_{\mathrm{out}}=\begin{pmatrix}v_x^{\mathrm{out}}&0\\0&v_p^{\mathrm{out}}\end{pmatrix},
\label{eq:output_moments}
\end{equation}
with
\begin{align}
v_x^{\mathrm{out}}&=\frac{\eta}{2}e^{-2r}+(1-\eta)n,\\
v_p^{\mathrm{out}}&=\frac{\eta}{2}e^{2r}+(1-\eta)n.
\label{eq:vxp_out}
\end{align}
Unless explicitly stated otherwise, all finite-distance comparisons below evaluate the structural separation between the received state $\rho_{\mathrm{out}}$ and the identical thermal reference $\rho_{\mathrm{ref}}$.

\subsection{Overlap and local-sensitivity benchmarks}
\label{sec:overlap}

For two quantum states $\rho_0$ and $\rho_1$, the squared Uhlmann fidelity is
\begin{equation}
F(\rho_0,\rho_1)=\left(\Trc\sqrt{\sqrt{\rho_0}\rho_1\sqrt{\rho_0}}\right)^2,
\label{eq:fidelity_def_section2}
\end{equation}
and we use the fidelity exponent
\begin{equation}
\xi_F:=-\ln F.
\label{eq:fidelity_exponent_section2}
\end{equation}
Fidelity and Chernoff-type quantities have established meanings in quantum-state comparison and discrimination \cite{Helstrom1976,Holevo2011,Jozsa1994,Audenaert2007,Banchi2015}. We use fidelity as an overlap-based benchmark, not as a quantity to be numerically out-ranged by transport.

For a parametric classical output law $p(x|\lambda)$, the classical Fisher information is
\begin{equation}
I_C(\lambda):=\int p(x|\lambda)\left[\partial_\lambda\ln p(x|\lambda)\right]^2dx,
\label{eq:cfi_definition}
\end{equation}
when the regularity conditions for this expression hold. The QFI $I_Q(\lambda)$ is the supremum of the CFI over quantum measurements and sets the quantum Cram\'er--Rao lower bound on local estimation variance \cite{Giovannetti2011,Demkowicz2015,Pinel2013}. These are \emph{local} quantities: they characterize infinitesimal changes of a parameterized family around a given operating point.

By contrast, the transport quantities introduced in the next sections are finite-distance comparisons to $\rho_{\mathrm{ref}}$. Their raw numerical scale depends on the ground cost, i.e., the cost assigned to moving probability mass between two outcome points. Multiplying that cost by a constant changes the numerical span of the resulting distance without creating a sensing advantage. We therefore do not use the numerical range of a transport plot as a performance criterion. The strong-loss comparison with fidelity is revisited quantitatively in Sec.~\ref{sec:isotropic} and Appendix~\ref{app:fading_asymptotics}.

% =================================================
\section{Isotropic Gaussian transport}
\label{sec:isotropic}
% =================================================

\subsection{Moment-space metric for the thermal-loss model}
\label{sec:isotropic_definition}

Because a Gaussian state is uniquely determined by $(\nu,V)$, the moment-space construction below is a genuine metric on the Gaussian-state family considered here. It should be distinguished from quantum Wasserstein metrics defined directly on general density operators.

We now introduce the first level of our hierarchy: an isotropic Gaussian moment-space transport distance. We associate each Gaussian state with the classical Gaussian law $\mathcal N(\nu,V)$ on $\mathbb R^2$ and define the distance by the quadratic Wasserstein distance between these associated laws. For two states $\rho_0$ and $\rho_1$ with respective moments $(\nu_0, V_0)$ and $(\nu_1, V_1)$, we define
\begin{eqnarray}
\Wiso(\rho_0,\rho_1)^2
&:=  &
\|\nu_0-\nu_1\|^2 +\nonumber \\
& 
\Trc&\!\left(
V_0+V_1
-
2\bigl(V_0^{1/2}V_1V_0^{1/2}\bigr)^{1/2}
\right)\negmedspace.
\label{eq:diso_def}
\end{eqnarray}
The covariance contribution is the standard Bures--Wasserstein (or Gelbrich) distance for positive definite matrices \cite{Gelbrich1990,Bhatia2019,Villani2009}. In the present work, we use this as a distance on the manifold of Gaussian states through their complete first- and second-moment parametrization. It is not presented as a universal quantum Wasserstein distance on arbitrary density operators \cite{Chen2018,DePalmaTrevisan2021}.

In the single-mode scenario considered here, the received state and the reference state both exhibit diagonal covariance matrices in the $(x,p)$ quadrature basis. The distance therefore simplifies to the sum of classical Wasserstein distances for the independent components:
\begin{equation}
\Wiso(\rho_0,\rho_1)^2
=
\|\nu_0-\nu_1\|^2
+
\sum_{k \in \{x,p\}}
\left(
\sqrt{v_k^{(0)}} - \sqrt{v_k^{(1)}}
\right)^2.
\label{eq:diso_diag}
\end{equation}

Applying this to the comparison between the received probe $\rho_{\out}$ and the background $\rho_{\refstate}$, and letting $n:=\nth+1/2$ denote the reference variance, we obtain the explicit formula:
\begin{eqnarray}
\Wiso(\rho_{\refstate},\rho_{\out})^2
&=&
2\etal\Ndisp
+
\left(\sqrt{v_x^{\out}}-\sqrt{n}\right)^2
\nonumber \\
& &+\left(\sqrt{v_p^{\out}}-\sqrt{n}\right)^2.
\label{eq:diso_explicit}
\end{eqnarray}
This expression separates the two contributions: the first term is purely displacement-driven, while the remaining two terms quantify the mismatch between the output covariance ellipse and the isotropic thermal reference.

Because both quadratures are treated symmetrically, \(\Wiso\) measures a global transport cost in phase space rather than the distinguishability relevant to a specific measurement quadrature. Reducing noise along one quadrature via squeezing does not necessarily lower the isotropic transport cost, as the distance includes the associated anti-squeezed deformation with the same coordinate weight along the conjugate variable. In this sense, isotropic transport geometry should be understood as a global moment-space deformation relative to the thermal reference, rather than as a surrogate for a measurement-dependent signal-to-noise ratio.

To compare the finite transport contrast with an overlap-based benchmark without conflating scale with performance, we rescale the fidelity exponent by its leading strong-loss coefficient. Appendix \ref{app:fading_asymptotics} gives, for nonzero coherent displacement,
\begin{equation}
 \Wiso^2=2\Ndisp\eta+O(\eta^2),\quad -\ln F=\frac{\Ndisp}{2n}\eta+O(\eta^2).
\end{equation}
Thus $\Wiso^2/[ -\ln F]\to4n$ as $\eta\to0$. Figure \ref{fig:isotropic_vs_fidelity} therefore compares $d_{\mathrm{iso}}^2$ with $4n[-\ln F]$, which removes the trivial leading scale difference. Deviations between the curves at finite transmissivity then reflect different geometric weightings rather than a claim that a larger raw numerical range is intrinsically more useful.

% =================================================
%\section{Analytical Limits of Isotropic Optimization}

\begin{figure}[htbp]
    \centering
    \includegraphics[width=0.78\linewidth]{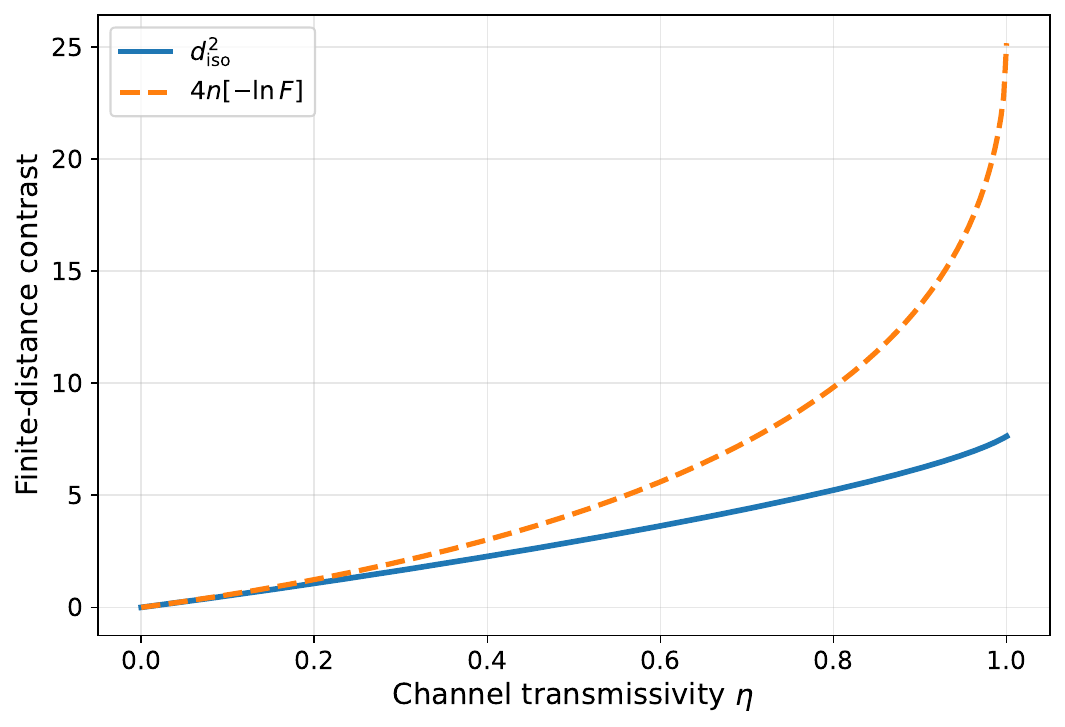}
\caption{\small Finite-distance comparison under thermal loss ($\Ntot=5$, $\beta=0.5$, $\nth=2$). The fidelity exponent is rescaled by $4n$, with $n=\nth+1/2$, so that $4n[-\ln F]$ and $\Wiso^2$ have the same leading displaced contribution as $\eta\to0$. The comparison removes the trivial leading scale factor; deviations at finite $\eta$ reflect different finite-distance weightings.}
\label{fig:isotropic_vs_fidelity}
\end{figure}

\subsection{Energy allocation at fixed total input energy}
\label{sec:nonoptimality}

We now investigate how the isotropic metric responds when a fraction $\beta\in[0,1]$ of a fixed total input energy $\Ntot$ is reallocated from coherent displacement to quadrature squeezing. This parametrization smoothly interpolates between a purely coherent probe ($\beta=0$) and a displaced squeezed probe ($\beta>0$) \cite{Braunstein2005,Weedbrook2012,Serafini2017}. We ask the following question: does allocating an infinitesimal amount of energy to squeezing improve or degrade the isotropic finite-distance contrast defined above?

To answer this, we evaluate the sensitivity of the explicit metric to the squeezing fraction $\beta$ in the vicinity of $\beta=0^+$. As detailed in Appendix \ref{app:proof_nonoptimality}, substituting the output moments into the metric and performing a local expansion yields the exact analytical derivative
\begin{equation}
\frac{\partial}{\partial \beta} \Wiso(\rho_{\refstate},\rho_{\out})^2 \Big|_{\beta=0^+} = \Ntot \etal \sqrt{n}\, s_0^{-3/2} \left( \frac{\etal}{2}-2s_0 \right),
\end{equation}
where $n = \nth + 1/2$ is the thermal reference variance and $s_0 := \frac{\etal}{2} + (1-\etal)n$ is the variance of the received state in the absence of squeezing.

\begin{proposition}[Local non-optimality of squeezing in the isotropic model]
\label{prop:local_nonoptimality}
For the isotropic thermal-reference model with vacuum variance $1/2$, introducing a small amount of squeezing strictly degrades the isotropic transport metric. Specifically, for any physical choice of parameters $\etal\in(0,1]$, $\nth\ge 0$, and $\Ntot>0$, we have
\begin{equation}
\frac{\partial}{\partial \beta} \Wiso(\rho_{\refstate},\rho_{\out})^2 \Big|_{\beta=0^+} < 0.
\end{equation}
\end{proposition}

\begin{proof}
The result follows directly from the exact derivative. The prefactor $\Ntot \etal \sqrt{n}\, s_0^{-3/2}$ is strictly positive. Expanding the parenthetical term yields
\begin{equation}
\frac{\etal}{2}-2s_0 = -\frac{\etal}{2}-2(1-\etal)n,
\end{equation}
which is strictly negative for all physical values of the channel parameters. Full mathematical details of the Taylor expansion are deferred to Appendix \ref{app:proof_nonoptimality}.
\end{proof}

This result should not be interpreted as a general statement against squeezing, which is well known to provide directional metrological advantages in quadrature-resolved settings, notably in homodyne detection and sub-shot-noise sensing \cite{Caves1981,Braunstein2005,Giovannetti2004,Giovannetti2011,Demkowicz2015}. Rather, it highlights that the isotropic metric evaluates a \emph{global} phase-space deformation relative to a symmetric thermal background. Because $\Wiso$ penalizes anti-squeezing exactly as much as it rewards squeezing, it does not automatically reflect the contrast of a specified quadrature measurement. The negative derivative is therefore not a pathology of the transport approach, but a structural signature of what global moment-space deformation actually entails.

The analytical derivative in Proposition \ref{prop:local_nonoptimality} establishes a local statement: the coherent boundary is strictly favored against an infinitesimal squeezing allocation for every physical parameter choice. Figure \ref{fig:isotropic_combined}(a) makes this decrease visible by plotting the finite penalty $\Wiso^2(0)-\Wiso^2(\beta)$ rather than displaying a numerically flat optimum curve. Figure \ref{fig:isotropic_combined}(b) shows the normalized right-derivative itself for several thermal occupancies, emphasizing the parameter-uniform sign proved analytically. These plots deliberately avoid promoting a finite-grid observation into a global theorem about all $\beta$.

\begin{figure}[htbp]
    \centering
    \begin{subfigure}[b]{0.48\textwidth}
        \centering
        \includegraphics[width=\linewidth]{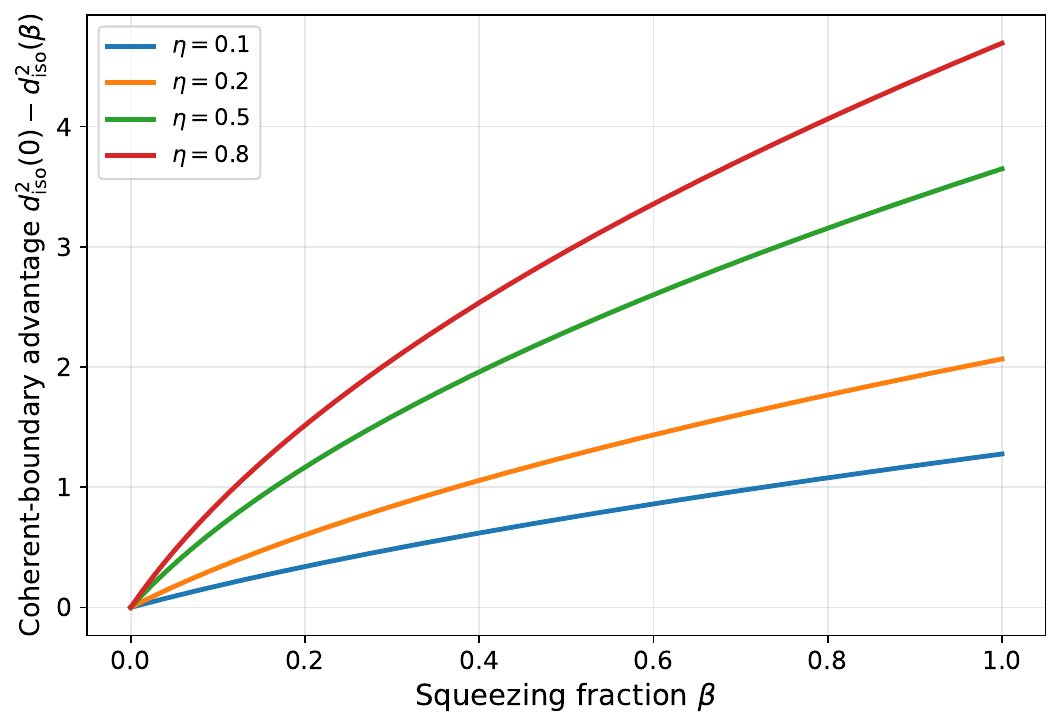}
        \caption{\small Finite coherent-boundary advantage for representative transmissivities.}
        \label{fig:isotropic_landscape}
    \end{subfigure}
    \hfill
    \begin{subfigure}[b]{0.48\textwidth}
        \centering
        \includegraphics[width=\linewidth]{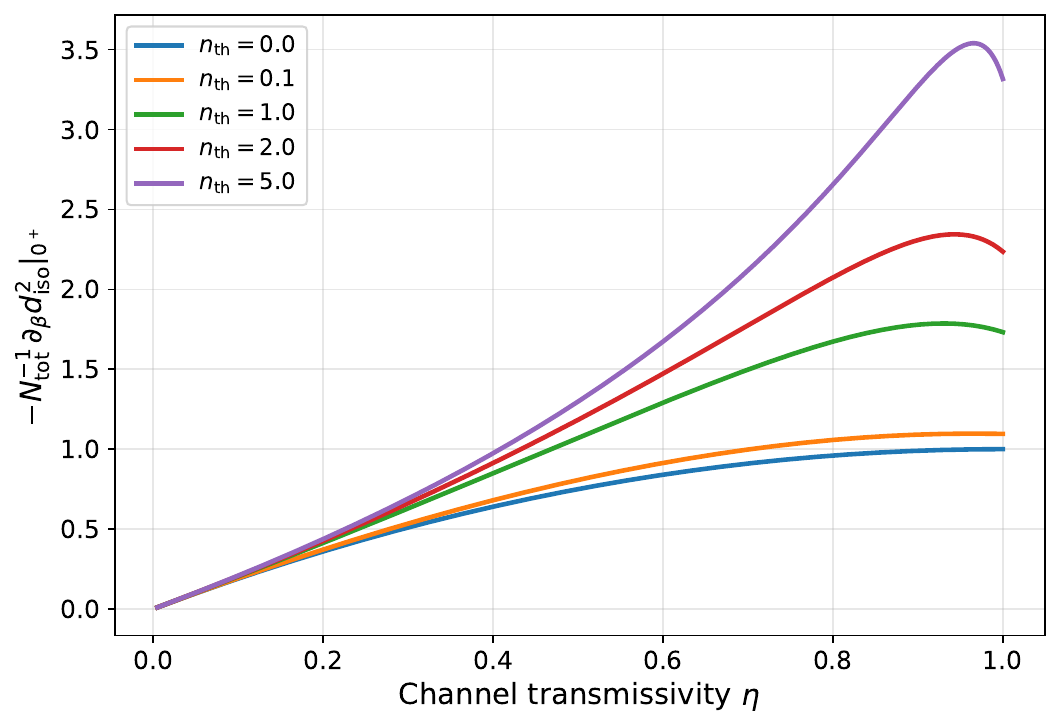}
        \caption{\small Normalized local slope at $\beta=0^+$ for several thermal occupancies.}
        \label{fig:isotropic_parametric}
    \end{subfigure}
    \caption{\small Revised numerical representation of Proposition \ref{prop:local_nonoptimality}. Panel (a) resolves the loss of isotropic contrast when energy is transferred from displacement to squeezing for $\Ntot=10$ and $\nth=0.1$. Panel (b) displays the strictly positive quantity $-\Ntot^{-1}\partial_\beta \Wiso^2|_{0^+}$, whose positivity is equivalent to the analytical local non-optimality of squeezing.}
    \label{fig:isotropic_combined}
\end{figure}

Global isotropic transport and measurement-specific contrast represent distinct geometric concepts. To resolve the contribution of a selected quadrature, the transport cost must reflect the measurement being performed. This motivates the next level of our hierarchy: a transport distance projected directly onto the observed quadrature. Whether that measurement is metrologically optimal is a separate Fisher-information question addressed in Sec.~\ref{sec:fisher}.

% =================================================
%\section{Projected Transport Geometry for Quadrature-Resolved Sensing}
% =================================================
\section{From state-space to measurement-level transport}
\label{sec:measurement_hierarchy}
% =================================================
\subsection{Homodyne-projected transport}
\label{sec:projected}

The isotropic distance compares both quadratures simultaneously and therefore does not encode which quadrature is actually observed. In continuous-variable architectures, particularly homodyne receivers, the measurement statistics are selected by the local-oscillator phase \cite{Braunstein2005,Caves1981,Leonhardt2008,Paris2004}. We therefore introduce a projected transport distance: the received state is compared with the reference background through the classical probability distribution induced by a designated quadrature measurement. This construction measures finite output contrast; it is not, by itself, a statement about optimal parameter-estimation sensitivity.

For any observation angle $\theta\in[0,2\pi)$, we define the rotated quadrature operator
\begin{equation}
X_\theta = x\cos\theta + p\sin\theta,
\label{eq:quadrature_theta}
\end{equation}
with the associated unit vector $e_\theta = (\cos\theta, \sin\theta)^T$. For a Gaussian state $\rho$ with moment vector $\nu$ and covariance matrix $V$, measuring this quadrature yields a one-dimensional normal distribution
\begin{equation}
P_\rho^{(\theta)} = \mathcal{N}\bigl(m_\theta(\rho),\, v_\theta(\rho)\bigr),
\label{eq:quadrature_distribution}
\end{equation}
characterized by the scalar mean and variance
\begin{equation}
m_\theta(\rho) = e_\theta^T \nu, \qquad v_\theta(\rho) = e_\theta^T V e_\theta.
\label{eq:quadrature_statistics}
\end{equation}

We then define the projected transport metric as the classical quadratic Wasserstein distance $W_2$ on the real line between these measurement-induced distributions \cite{Villani2009}:
\begin{equation}
\Wtheta(\rho_0,\rho_1)^2
:=
W_2\!\left(P_{\rho_0}^{(\theta)},\,P_{\rho_1}^{(\theta)}\right)^2.
\label{eq:dtheta_def}
\end{equation}
For one-dimensional Gaussian laws, this distance admits the closed-form expression
\begin{eqnarray}
\Wtheta(\rho_0,\rho_1)^2
&=&
\left(m_\theta(\rho_0)-m_\theta(\rho_1)\right)^2 \nonumber \\
& &+
\left(\sqrt{v_\theta(\rho_0)}-\sqrt{v_\theta(\rho_1)}\right)^2.
\label{eq:dtheta_gaussian}
\end{eqnarray}
Unlike the isotropic distance $\Wiso$, the projected metric $\Wtheta$ depends explicitly on the observation angle, making it a geometry natively adapted to homodyne sensing protocols. More precisely, $W_2$ is a metric on the one-dimensional output distributions $P_\rho^{(\theta)}$. Its pullback $\Wtheta$ to quantum states is generally only a pseudometric because two distinct states can produce the same marginal distribution for a fixed quadrature.

%%%%%%%%%%%%%%%%%%%%%%

We now apply this framework to the lossy squeezed-displaced Gaussian model. Using the output and reference moments defined in Sec.~\ref{sec:model}, the projected statistics along an arbitrary observation angle $\theta$ are straightforward to evaluate. The mean displacements project as
\begin{equation}
m_\theta(\rho_{\out}) = \sqrt{2\etal\Ndisp}\cos\theta, \qquad m_\theta(\rho_{\refstate}) = 0,
\end{equation}
while the variances, owing to the diagonal structure of the covariance matrices in the $(x,p)$ basis, become
\begin{equation}
v_\theta(\rho_{\out}) = v_x^{\out}\cos^2\theta + v_p^{\out}\sin^2\theta, \quad v_\theta(\rho_{\refstate}) = n.
\end{equation}

Substituting these classical statistics into (\ref{eq:dtheta_gaussian}) yields the explicit projected transport metric:
\begin{multline}
\Wtheta(\rho_{\refstate},\rho_{\out})^2
=
2\etal\Ndisp\cos^2\theta + \\
\left(
\sqrt{v_x^{\out}\cos^2\theta+v_p^{\out}\sin^2\theta}
-\sqrt{n}
\right)^2.
\label{eq:dtheta_explicit}
\end{multline}

Crucially, this explicit form demonstrates how the choice of measurement quadrature actively reshapes the transport geometry seen by the receiver.

The departure from the isotropic framework is immediate. Whereas $\Wiso$ combines the two conjugate covariance deformations symmetrically, $\Wtheta$ isolates the geometry of the measured marginal. The finite transport contrast can therefore depend strongly on the local-oscillator phase and on the allocation between displacement and squeezing. In particular, a large contrast on the anti-squeezed axis reflects a large variance deformation, not a reduction of measurement noise. Accordingly, $\Wtheta$ should not be conflated with signal-to-noise ratio or Fisher information; it remains a Wasserstein contrast between classical output distributions.

A useful analytical property of (\ref{eq:dtheta_explicit}) is that it makes the optimization over the measurement quadrature analytically tractable. 

\begin{proposition}[Boundary optimization of the projected metric]
\label{prop:boundary_quadrature}
For any fixed set of state parameters $(\eta,\Ndisp,r,n)$, the squared projected transport metric $\Wtheta(\rho_{\refstate},\rho_{\out})^2$ is a convex function of $\cos^2\theta$. Consequently, the optimal measurement quadrature is always aligned with one of the principal axes of the output covariance ellipse:
\begin{equation}
\theta_{\mathrm{opt}} \in \left\{0, \frac{\pi}{2}\right\}.
\end{equation}
Equivalently, the global continuous search reduces to a discrete binary choice:
\begin{eqnarray}
\max_{\theta} \Wtheta(\rho_{\refstate},\rho_{\out})^2
&=
&\max\!\left\{
\Wtheta(\rho_{\refstate},\rho_{\out})^2 \Big|_{\theta=0}, \right. \nonumber \\
& &\left.\Wtheta(\rho_{\refstate},\rho_{\out})^2 \Big|_{\theta=\pi/2}
\right\}.
\end{eqnarray}

\end{proposition}

\begin{proof}
Introducing the variable $u := \cos^2\theta \in [0,1]$ and noting that $\sin^2\theta = 1-u$, we can rewrite the projected metric using the output variance difference $\Delta := v_p^{\out} - v_x^{\out} \ge 0$. This yields
\begin{equation}
\Wtheta^2(u) = 2\eta\Ndisp u + \left( \sqrt{v_p^{\out} - \Delta u} - \sqrt{n} \right)^2.
\end{equation}
Differentiating twice with respect to $u$ gives
\begin{equation}
\frac{d^2}{du^2} \Wtheta^2(u) = \frac{\Delta^2\sqrt{n}}{2(v_p^{\out} - \Delta u)^{3/2}} \ge 0.
\end{equation}
Because the second derivative is non-negative, $\Wtheta^2(u)$ is convex on $[0,1]$. By the fundamental properties of convex functions, its maximum over a compact interval is necessarily attained at one of the boundaries. These endpoints, $u=1$ and $u=0$, correspond respectively to $\theta=0$ (the displaced quadrature) and $\theta=\pi/2$ (the orthogonal anti-squeezed quadrature).
\end{proof}

Proposition \ref{prop:boundary_quadrature} gives the projected contrast a simple receiver-setting interpretation: maximizing the transport contrast over the observation angle reduces strictly to a boundary choice between the principal axes of the output covariance ellipse. Consequently, the optimal homodyne receiver setting can be identified without any continuous angular search. One defines the deterministic quadrature-selection rule
\begin{eqnarray}
\theta_{\mathrm{sel}}
&:=&
\arg\max\!\left\{
\Wtheta(\rho_{\refstate},\rho_{\out})^2\Big|_{\theta=0}, \right. \nonumber \\
& &\left.\Wtheta(\rho_{\refstate},\rho_{\out})^2\Big|_{\theta=\pi/2}
\right\}.
\label{eq:theta_sel_rule}
\end{eqnarray}

Evaluating the metric at these two candidate boundaries yields
\begin{equation}
\Wtheta(\rho_{\refstate},\rho_{\out})^2\Big|_{\theta=0}
=
2\eta\Ndisp+\left(\sqrt{v_x^{\out}}-\sqrt n\right)^2,
\label{eq:wtheta_theta0}
\end{equation}
and
\begin{equation}
\Wtheta(\rho_{\refstate},\rho_{\out})^2\Big|_{\theta=\pi/2}
=
\left(\sqrt{v_p^{\out}}-\sqrt n\right)^2.
\label{eq:wtheta_thetapi2}
\end{equation}
The optimal quadrature is thus dictated by a direct competition: the coherent displacement natively favors $\theta=0$, whereas the noise covariance deformation (driven by anti-squeezing) may favor the orthogonal axis $\theta=\pi/2$. More explicitly, the displacement-aligned axis is optimal ($\theta_{\mathrm{opt}}=0$) if and only if
\begin{equation}
2\eta\Ndisp
\ge
\bigl(\sqrt{v_p^{\out}}-\sqrt{v_x^{\out}}\bigr)
\bigl(\sqrt{v_p^{\out}}+\sqrt{v_x^{\out}}-2\sqrt n\bigr).
\label{eq:theta_opt_condition}
\end{equation}
Otherwise, the receiver must be switched to $\theta_{\mathrm{opt}}=\pi/2$. 

The practical consequence of this binary choice becomes particularly clear in the strongly lossy regime ($\eta \to 0$). Evaluating the difference between the two candidate axes yields
\begin{eqnarray}
\Wtheta(\rho_{\refstate},\rho_{\out})^2\Big|_{\theta=0} &-& \Wtheta(\rho_{\refstate},\rho_{\out})^2\Big|_{\theta=\pi/2} \nonumber \\
&=&
2\eta\Ndisp + O(\eta^2).
\end{eqnarray}
Consequently, as long as the probe carries some coherent amplitude ($\Ndisp > 0$), the optimal measurement quadrature is strictly aligned with the displaced axis for sufficiently small transmissivities. In this severe-loss limit, the noise deformation induced by squeezing is entirely washed out by the thermal background, dictating that the homodyne detector must track the surviving coherent displacement to maximize this projected finite-distance contrast.

We next inspect the joint finite-distance objective $(\beta,\theta)\mapsto \Wtheta^2$. Figure~\ref{fig:projected_angle} shows a strong angular dependence for the representative parameters $\Ntot=10$ and $\nth=0.1$. Near the displaced/squeezed $x$ axis the finite contrast is displacement dominated, whereas near $\theta=\pi/2$ the anti-squeezed variance deformation can dominate. For the parameter set shown, the numerical optimum switches between the two energy-allocation boundaries $\beta=0$ and $\beta=1$. We do not promote this boundary behavior in $\beta$ to a general theorem; the analytical boundary propositions concern optimization over the quadrature angle at fixed state parameters.

\begin{figure}[htbp]
\centering
\begin{subfigure}[b]{0.48\textwidth}
\centering
\includegraphics[width=\linewidth]{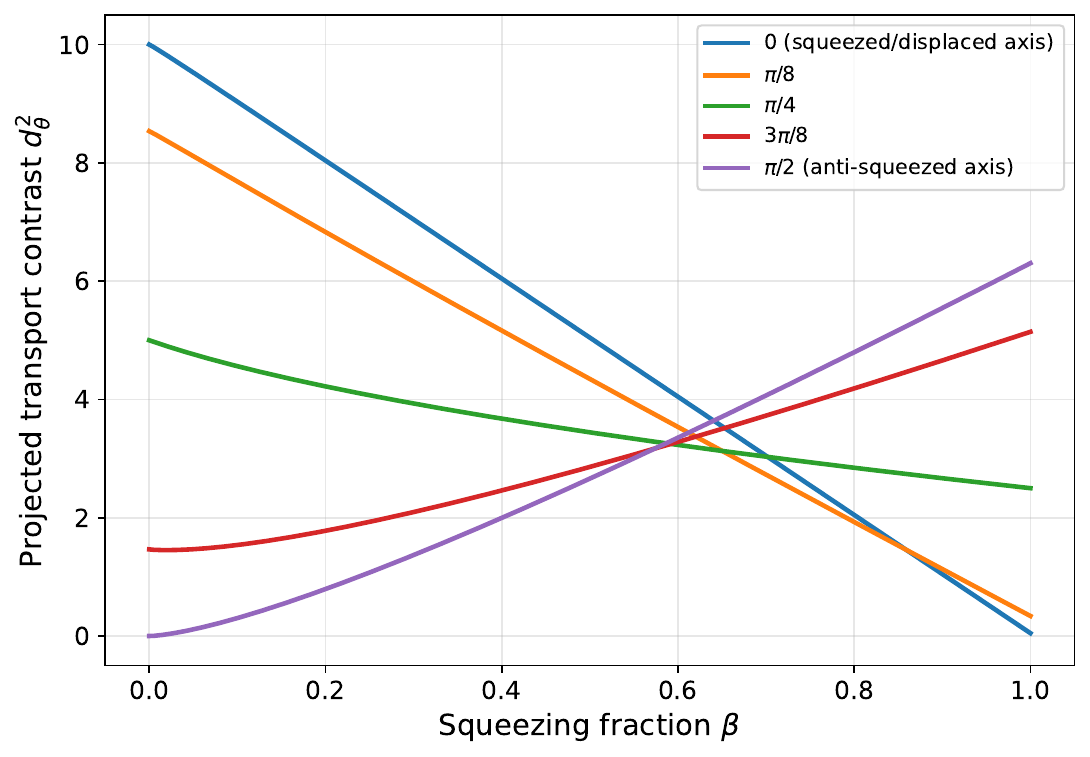}
\caption{\small Projected contrast versus energy allocation at $\eta=0.5$.}
\end{subfigure}
\hfill
\begin{subfigure}[b]{0.48\textwidth}
\centering
\includegraphics[width=\linewidth]{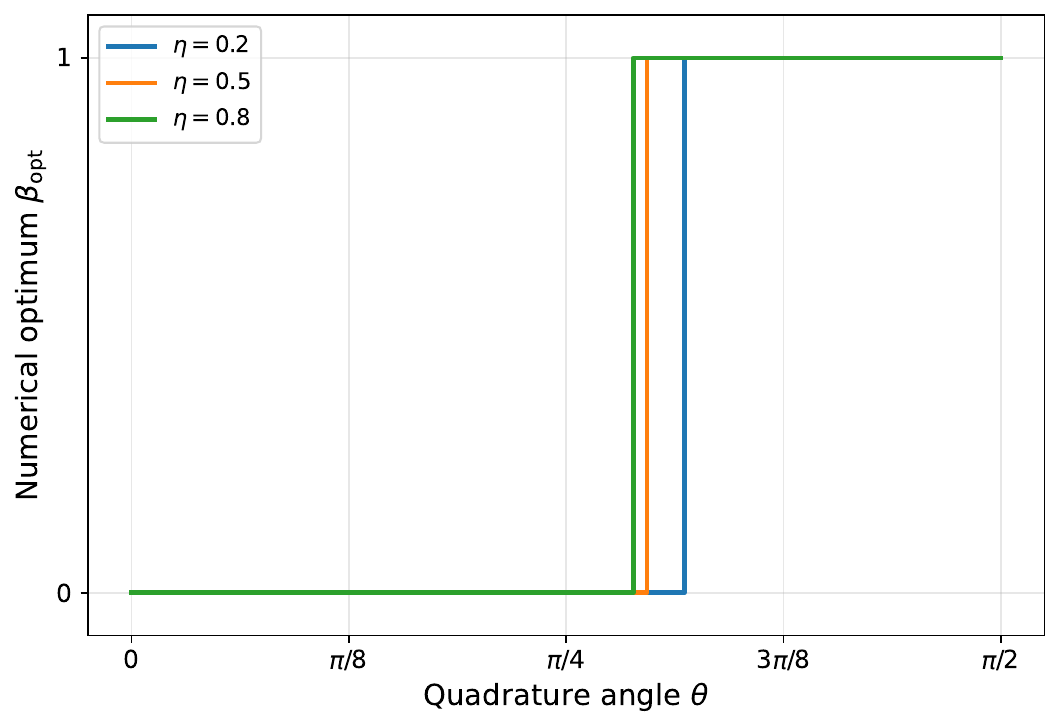}
\caption{\small Numerical optimum for the parameters shown.}
\end{subfigure}
\caption{\small Measurement dependence of the projected transport contrast for $\Ntot=10$ and $\nth=0.1$. With our convention $v_x=\tfrac12e^{-2r}$ and $v_p=\tfrac12e^{2r}$, $\theta=0$ is the squeezed/displaced axis whereas $\theta=\pi/2$ is the anti-squeezed axis. Panel (a) shows that the transport contrast can favor either displacement or covariance deformation depending on the measured quadrature. Panel (b) reports the resulting boundary switch for the displayed parameter set; it is a numerical observation for this set, not a general theorem in the squeezing fraction.}
\label{fig:projected_angle}
\end{figure}

The projected construction describes ideal homodyne statistics. We now include detector noise explicitly and define the final level of the hierarchy directly on the measured output law.

% =================================================

\subsection{Measurement-aware transport}
\label{sec:measurement}

The projected metric introduced in Sec.~\ref{sec:projected} moves the comparison from moment space to a specified classical output law: state distinguishability is no longer evaluated through a global phase-space cost, but through the classical distributions induced by a specific measurement. The same idea extends naturally when the detector itself adds noise or classical post-processing.

We formalize this principle by making the measurement chain explicit from the outset. Let $\mathcal{M}$ denote a generic measurement procedure acting on a quantum state $\rho$ to produce a classical outcome variable. Depending on the experimental architecture, $\mathcal{M}$ may encompass the choice of observable or positive operator-valued measure (POVM), local-oscillator tuning, detector inefficiencies, additive electronic noise, and subsequent classical post-processing \cite{Leonhardt2008,Paris2004,Helstrom1976,Holevo2011}. Regardless of these underlying complexities, the measurement chain maps any input state $\rho$ to a classical probability distribution $P_\rho^{\mathcal M}$ over the outcome space.

The definition below is made at the classical-output level and therefore does not require a phase-space representation of the underlying state. It provides a measurement-specific finite-distance contrast based entirely on accessible output statistics; it is not asserted to be a universal state-space distinguishability measure.

\begin{definition}[Measurement-aware transport distance]
Let $\rho_0$ and $\rho_1$ be two quantum states, and let $\mathcal{M}$ be an arbitrary measurement chain yielding the classical output distributions $P_{\rho_0}^{\mathcal M}$ and $P_{\rho_1}^{\mathcal M}$. We define the measurement-aware transport distance between the two detector-output laws as
\begin{equation}
\dM(\rho_0,\rho_1)
:=
W_2\!\left(P_{\rho_0}^{\mathcal M},\,P_{\rho_1}^{\mathcal M}\right),
\label{eq:dM_def}
\end{equation}
where $W_2$ denotes the classical quadratic Wasserstein distance on the corresponding outcome space \cite{Villani2009}.
\end{definition}

The Wasserstein distance in Eq. (\ref{eq:dM_def}) is a genuine metric on the classical detector-output distributions. Its pullback through a fixed measurement map is generally only a pseudometric on quantum states, because distinct states can yield identical output laws. We therefore interpret $\dM$ as a measurement-specific finite-distance contrast defined after the quantum-to-classical interface.

Our previous constructions arise naturally as special cases of this framework. For ideal homodyne detection without additional detector noise, $\dM$ reduces to the projected distance $\Wtheta$. With an explicit detector model, it is the final level of the present hierarchy:
\begin{equation}
\Wiso \;\longrightarrow\; \Wtheta \;\longrightarrow\; \dM.
\end{equation}
This progression tracks the level of physical description, moving from global phase-space transport, to quadrature-resolved geometry, and finally to full measurement-aware statistics.

For the Gaussian moment geometry used here, the three levels are related by an actual contraction statement. Associate to a Gaussian state with moments $(\nu,V)$ the classical Gaussian law $\mathcal N(\nu,V)$ on $\mathbb R^2$. The isotropic quantity $\Wiso$ is the quadratic Wasserstein distance between these laws. Projection onto $X_\theta=e_\theta^T X$ is a 1-Lipschitz map, so Wasserstein contraction under push-forward gives $\Wtheta\le \Wiso$. If the same additive detector noise is then applied to both projected laws, common-noise coupling yields a second contraction.

\begin{proposition}[Contraction through projection and common detector noise]
\label{prop:transport_contraction}
For the Gaussian model of Eqs. (\ref{eq:dtheta_def}) and (\ref{eq:dM_gaussian_noisy}),
\begin{equation}
\dM(\rho_0,\rho_1)\le \Wtheta(\rho_0,\rho_1)\le \Wiso(\rho_0,\rho_1).
\label{eq:transport_contraction}
\end{equation}
\end{proposition}
\begin{proof}
Let $\mu_i=\mathcal N(\nu_i,V_i)$ be the Gaussian laws associated with the moments of $\rho_i$. The map $T_\theta(x)=e_\theta^T x$ has Lipschitz constant one. Therefore, for every coupling $\pi$ of $\mu_0$ and $\mu_1$, the pushed-forward coupling satisfies $|T_\theta(x)-T_\theta(y)|^2\le\|x-y\|^2$, and minimization gives $W_2(T_{\theta\#}\mu_0,T_{\theta\#}\mu_1)\le W_2(\mu_0,\mu_1)$. This is $\Wtheta\le \Wiso$. For additive detector noise, let $Z$ be an independent zero-mean Gaussian variable of variance $\sigma_d^2$ and use the same $Z$ in an optimal coupling $(X,Y)$ of the projected laws. Then $|(X+Z)-(Y+Z)|^2=|X-Y|^2$, so minimization over all output couplings gives $\dM\le \Wtheta$.
\end{proof}

An analytically tractable, yet highly practical, example is a homodyne receiver degraded by additive Gaussian electronic noise of variance $\sigma_d^2$. The measurement map then amounts to a convolution $P_\rho^{\mathcal M} = P_\rho^{(\theta)} * \mathcal{N}(0,\sigma_d^2)$. This classical noise floor preserves the Gaussian character of the distributions, leaving the mean displacement unchanged while regularizing the variance ($v_\theta \mapsto v_\theta + \sigma_d^2$). The measurement-aware transport metric therefore becomes
\begin{multline}
\dM(\rho_0,\rho_1)^2
=
\bigl(m_\theta(\rho_0)-m_\theta(\rho_1)\bigr)^2 \\
+
\left(
\sqrt{v_\theta(\rho_0)+\sigma_d^2}
-
\sqrt{v_\theta(\rho_1)+\sigma_d^2}
\right)^2.
\label{eq:dM_gaussian_noisy}
\end{multline}

Applying this framework to our lossy squeezed-displaced model is now straightforward. Substituting the classical statistics $m_\theta$ and $v_\theta$ previously derived in Sec.~\ref{sec:projected}, we immediately obtain the explicit noisy transport metric:
\begin{multline}
\dM(\rho_{\refstate},\rho_{\out})^2
=
2\eta\Ndisp\cos^2\theta \\
+
\left(
\sqrt{v_x^{\out}\cos^2\theta+v_p^{\out}\sin^2\theta+\sigma_d^2}
-
\sqrt{n+\sigma_d^2}
\right)^2.
\label{eq:dM_explicit_noisy}
\end{multline}
Equation (\ref{eq:dM_explicit_noisy}) generalizes the ideal projected metric to realistic detection chains. Remarkably, the underlying structure of the projected geometry survives this noise regularization: as we now show, the continuous angular search for the optimal receiver setting still collapses to a strict boundary optimization.

\begin{proposition}[Boundary optimization of the noisy projected metric]
\label{prop:boundary_quadrature_noisy}
For any fixed set of state parameters $(\eta,\Ndisp,r,n)$ and detector noise variance $\sigma_d^2\ge 0$, the squared measurement-aware metric $\dM(\rho_{\refstate},\rho_{\out})^2$ is a convex function of $\cos^2\theta$. Consequently, the optimal measurement quadrature remains aligned with one of the principal axes of the output covariance ellipse:
\begin{equation}
\theta_{\mathrm{opt}}\in\left\{0,\frac{\pi}{2}\right\}.
\end{equation}
Equivalently, the continuous angular optimization reduces to the discrete choice:
\begin{eqnarray}
\max_\theta \dM(\rho_{\refstate},\rho_{\out})^2
&=&
\max\!\left\{
\dM(\rho_{\refstate},\rho_{\out})^2\Big|_{\theta=0},\right. \nonumber \\
& &\left.\dM(\rho_{\refstate},\rho_{\out})^2\Big|_{\theta=\pi/2}
\right\}.
\end{eqnarray}
\end{proposition}

\begin{proof}
Introducing the variable $u := \cos^2\theta \in [0,1]$ (so that $\sin^2\theta = 1-u$) and the noise-regularized constants $B := v_p^{\out} + \sigma_d^2$ and $N := n + \sigma_d^2$, we can rewrite the metric using the variance difference $\Delta := v_p^{\out} - v_x^{\out} \ge 0$. This yields
\begin{equation}
\dM^2(u) = 2\eta\Ndisp u + \left( \sqrt{B - \Delta u} - \sqrt{N} \right)^2.
\end{equation}
Differentiating twice with respect to $u$ gives
\begin{equation}
\frac{d^2}{du^2} \dM^2(u) = \frac{\Delta^2\sqrt{N}}{2(B - \Delta u)^{3/2}} \ge 0.
\end{equation}
Because the second derivative is non-negative, $\dM^2(u)$ is convex on $[0,1]$. By the properties of convex functions, its maximum over this compact interval is necessarily attained at one of the boundaries: $u=1$ (corresponding to $\theta=0$) or $u=0$ (corresponding to $\theta=\pi/2$).
\end{proof}

As in the ideal projected case, the boundary result defines a binary quadrature-selection rule for maximizing this particular finite-distance contrast:
\begin{multline}
\theta_{\mathrm{sel}}^{\mathcal M}
:=
\arg\max\!\left\{
\dM(\rho_{\refstate},\rho_{\out})^2\Big|_{\theta=0},\right.\\
\left.\dM(\rho_{\refstate},\rho_{\out})^2\Big|_{\theta=\pi/2}
\right\}.
\label{eq:theta_sel_meas}
\end{multline}

Evaluating the metric at these two candidate axes yields
\begin{multline}
\dM(\rho_{\refstate},\rho_{\out})^2\Big|_{\theta=0}
=
2\eta\Ndisp
+\\
\left(
\sqrt{v_x^{\out}+\sigma_d^2}
-
\sqrt{n+\sigma_d^2}
\right)^2,
\label{eq:dM_theta0}
\end{multline}
and
\begin{equation}
\dM(\rho_{\refstate},\rho_{\out})^2\Big|_{\theta=\pi/2}
=
\left(
\sqrt{v_p^{\out}+\sigma_d^2}
-
\sqrt{n+\sigma_d^2}
\right)^2.
\label{eq:dM_thetapi2}
\end{equation}
Just as in the noise-free case, the optimal tuning is governed by a direct competition: the coherent displacement natively favors the aligned axis ($\theta=0$), while the covariance deformation favors the orthogonal quadrature ($\theta=\pi/2$). Detector noise does not change the convexity argument; in Eq.~(\ref{eq:dM_explicit_noisy}) it enters additively under the square roots. Increasing $\sigma_d^2$ reduces the covariance contribution to the finite-distance contrast, while the displacement term is unchanged in this additive-noise model.

The mathematical status of $\dM$ is therefore precise: it is the $W_2$ metric on the detector-output probability laws and, through a fixed measurement map, generally a pseudometric on quantum states. In the Gaussian homodyne model it is available in closed form; for other measurements the same definition can be applied whenever the output distributions and their ground cost are specified. No claim of statistical optimality follows from this definition alone.

With the hierarchy now defined, the next question is how a finite measurement-level contrast relates to the standard local sensitivity criteria used in quantum metrology.

% =================================================
\section{Relation to Fisher-information-based metrology}
\label{sec:fisher}
% =================================================

The preceding quantities compare outputs at two finite parameter values. To distinguish that task from local parameter estimation, let $\lambda$ parameterize a one-dimensional Gaussian homodyne law
\begin{equation}
p(x|\lambda)=\mathcal N(m_\lambda,v_\lambda).
\end{equation}
For neighboring values $\lambda$ and $\lambda+d\lambda$, the quadratic Wasserstein distance expands as
\begin{equation}
W_2^2(p_\lambda,p_{\lambda+d\lambda})=
\left[(\partial_\lambda m)^2+\frac{(\partial_\lambda v)^2}{4v}\right]d\lambda^2+o(d\lambda^2).
\label{eq:local_w2_metric}
\end{equation}
The coefficient in brackets is the local quadratic form induced by $W_2$ on this measured family. For the same homodyne statistics, Eq.~(\ref{eq:cfi_definition}) gives
\begin{equation}
I_C(\lambda)=\frac{(\partial_\lambda m)^2}{v}
+\frac12\frac{(\partial_\lambda v)^2}{v^2}.
\label{eq:classical_fisher_gaussian}
\end{equation}
The two local quadratic forms are not identical: Fisher information normalizes changes by the output variance, whereas Wasserstein transport retains the scale set by the outcome coordinate.

For completeness, the QFI used in the numerical comparison can also be written directly for a mixed one-mode Gaussian state with mean $\nu_\lambda$ and covariance $V_\lambda$. In our convention $V_{\mathrm{vac}}=I_2/2$, define the purity
\begin{equation}
P_\lambda:=\frac{1}{2\sqrt{\det V_\lambda}},
\end{equation}
and let a dot denote $\partial_\lambda$. The QFI is \cite{Pinel2013}
\begin{multline}
I_Q(\lambda)=
\frac{\Trc[(V^{-1}\dot V)^2]}{2(1+P^2)}
+\frac{2\dot P^{\,2}}{1-P^4}
+\dot\nu^{T}V^{-1}\dot\nu,
\label{eq:qfi_gaussian}
\end{multline}
with the corresponding limiting expression at pure-state boundary points. By construction, $I_C\le I_Q$ for any fixed measurement. In the numerical calculations we independently checked Eq.~(\ref{eq:qfi_gaussian}) against the local Bures curvature obtained from the Gaussian fidelity.

We now take $\lambda=\eta$ and optimize the probe energy split. Figure~\ref{fig:fisher_comparison} compares the squeezing fraction maximizing QFI, ideal-homodyne CFI, noisy-homodyne CFI, and the finite measurement-aware contrast to the thermal reference. In each classical curve the homodyne angle is optimized as well. For $\Ntot=10$ and $\nth=0.1$, QFI and CFI progressively favor a hybrid displacement--squeezing allocation as $\eta$ increases, whereas the displayed measurement-aware distance remains at the coherent boundary. Far from being a discrepancy, this divergence highlights the fundamental operational distinction between the two approaches: QFI and CFI quantify sensitivity to infinitesimal parameter variations, whereas $\dM$ quantifies a finite-distance contrast to the fixed background $\rho_{\mathrm{ref}}$.
%The different optima are expected because the objectives answer different questions: QFI/CFI optimize infinitesimal estimation of $\eta$, while $\dM$ maximizes a finite contrast to $\rho_{\mathrm{ref}}$.
\begin{figure}[htbp]
{\centering
\includegraphics[width=0.88\linewidth]{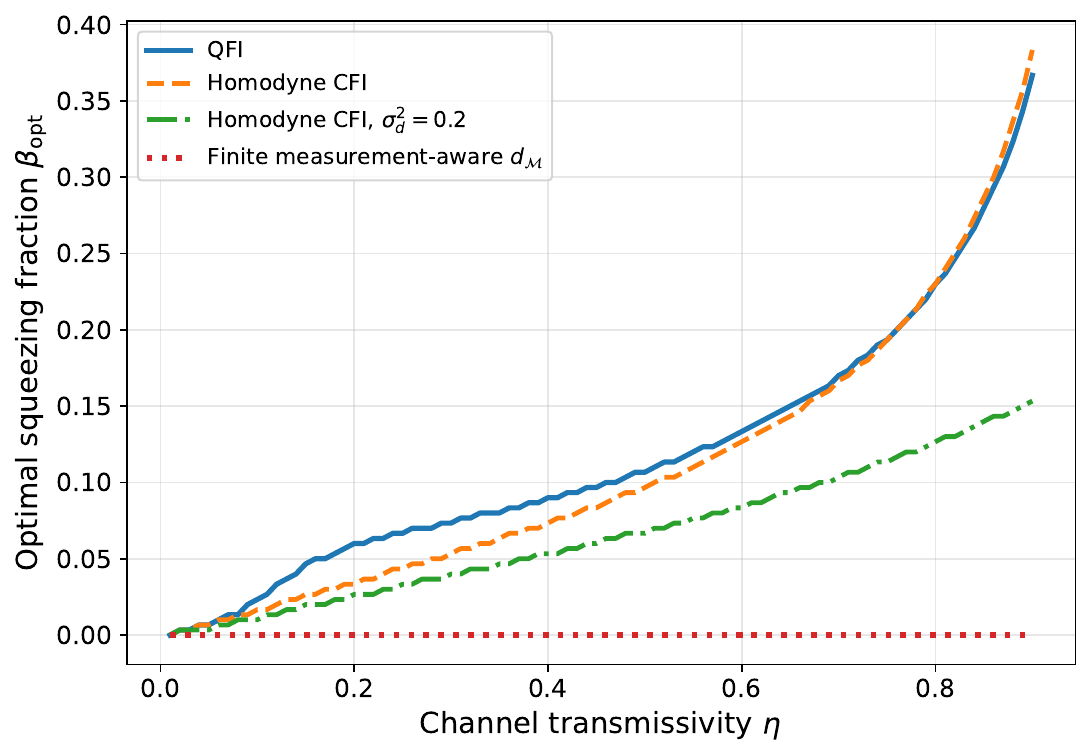}
}
\caption{\small Energy-allocation optima for different objectives ($\Ntot=10$, $\nth=0.1$). The QFI is measurement independent. The ideal-homodyne CFI is optimized over the homodyne angle. The noisy CFI and the measurement-aware distance $d_{\mathcal M}$ are computed with additive detector noise $\sigma_d^2=0.2$ and are also optimized over the homodyne angle. The resulting differences in $\beta_{\mathrm{opt}}$ reflect distinct local-estimation and finite-distance objectives, rather than inconsistent notions of sensitivity.}
\label{fig:fisher_comparison}
\end{figure}

For strong loss it is useful to distinguish intensity transmissivity $\eta$ from field-amplitude transmissivity $\tau:=\sqrt{\eta}$. Since the output displacement is linear in $\tau$, Fisher information with respect to $\eta$ can have a coordinate singularity at $\eta=0$ when displacement is present. The regular reparametrization satisfies
\begin{equation}
I_\tau=4\tau^2 I_\eta.
\end{equation}
We retain $\eta$ for the finite-distance channel plots because it is the standard intensity transmissivity, but interpret local Fisher behavior at complete loss through the regular coordinate $\tau$.

Having separated finite transport contrast from local metrological sensitivity for a fixed channel parameter, we next ask whether the measurement-aware contrast can serve as a directly estimated statistic when that same channel parameter varies in time.

% =================================================
\section{Application: finite-sample monitoring of a block-fading channel}
\label{sec:fading}
% =================================================

All previous sections condition on a fixed transmissivity $\eta$. We now use the same thermal-loss channel as a simple application in which the transmissivity varies between measurement blocks. By \emph{block fading} we mean that $\eta$ is constant during a block of $M$ detector samples and is redrawn independently for the next block. For illustration we choose the bounded distribution
\begin{equation}
\eta\sim\mathrm{Beta}(a_f,b_f),\qquad a_f=2,\quad b_f=5,
\label{eq:fading_law}
\end{equation}
not as a universal fading law but as a convenient asymmetric distribution on $[0,1]$ with substantial weight in the lossy regime.

The receiver does not know the realized block transmissivity. We use the $x$ homodyne quadrature and additive Gaussian detector noise of variance $\sigma_d^2$. Define
\begin{equation}
A:=\sqrt{2N_{\mathrm{disp}}},\qquad
m_\eta=A\sqrt{\eta},
\end{equation}
and
\begin{equation}
s_\eta:=v_x^{\mathrm{out}}(\eta)+\sigma_d^2
=\frac{\eta}{2}e^{-2r}+(1-\eta)n+\sigma_d^2.
\label{eq:fading_variance}
\end{equation}
Conditioned on one block, the $M$ detector outcomes are independent samples
\begin{equation}
X_k|\eta\sim\mathcal N(m_\eta,s_\eta),\qquad k=1,\ldots,M.
\label{eq:fading_samples}
\end{equation}
Thus the application differs from the earlier parameter sweeps: $\eta$ is hidden and the transport contrast must be formed from a finite measurement record.

Let
\begin{equation}
\widehat m:=\frac1M\sum_{k=1}^M X_k,
\qquad
\widehat s:=\frac{1}{M-1}\sum_{k=1}^M(X_k-\widehat m)^2
\label{eq:empirical_moments}
\end{equation}
be the sample mean and unbiased sample variance. The reference detector variance is $s_0=n+\sigma_d^2$. We define the empirical measurement-aware contrast
\begin{equation}
\widehat \dM^{\,2}
:=\widehat m^{\,2}+
\left(\sqrt{\widehat s}-\sqrt{s_0}\right)^2.
\label{eq:empirical_dM}
\end{equation}
Its noise-free calibration curve at the level of measurement statistics is
\begin{equation}
\dM^2(\eta)=m_\eta^2+
\left(\sqrt{s_\eta}-\sqrt{s_0}\right)^2.
\label{eq:dM_calibration}
\end{equation}
For the parameter set used below this function is monotone on $[0,1]$. 
We therefore define the \emph{$d_{\mathcal M}$-calibrated monitor} as the inverse-calibration estimator
\begin{equation}
\widehat\eta_{\mathcal{M}}:=
\left[\dM^2\right]^{-1}
\!\left(\widehat \dM^{\,2}\right),
\label{eq:eta_dM_monitor}
\end{equation}
with values clipped to the physical interval when sampling fluctuations fall outside the calibrated range. Here \emph{monitor} means only this scalar estimator constructed from the transport contrast; it does not imply statistical optimality.

Figure~\ref{fig:fading_calibration} shows the sampling fluctuations around Eq.~(\ref{eq:dM_calibration}) for $\Ntot=5$, $\beta=0.5$, $\nth=2$, $\sigma_d^2=0.5$, and $M=64$ samples per block.

\begin{figure}[htbp]
\centering
\includegraphics[width=0.88\linewidth]{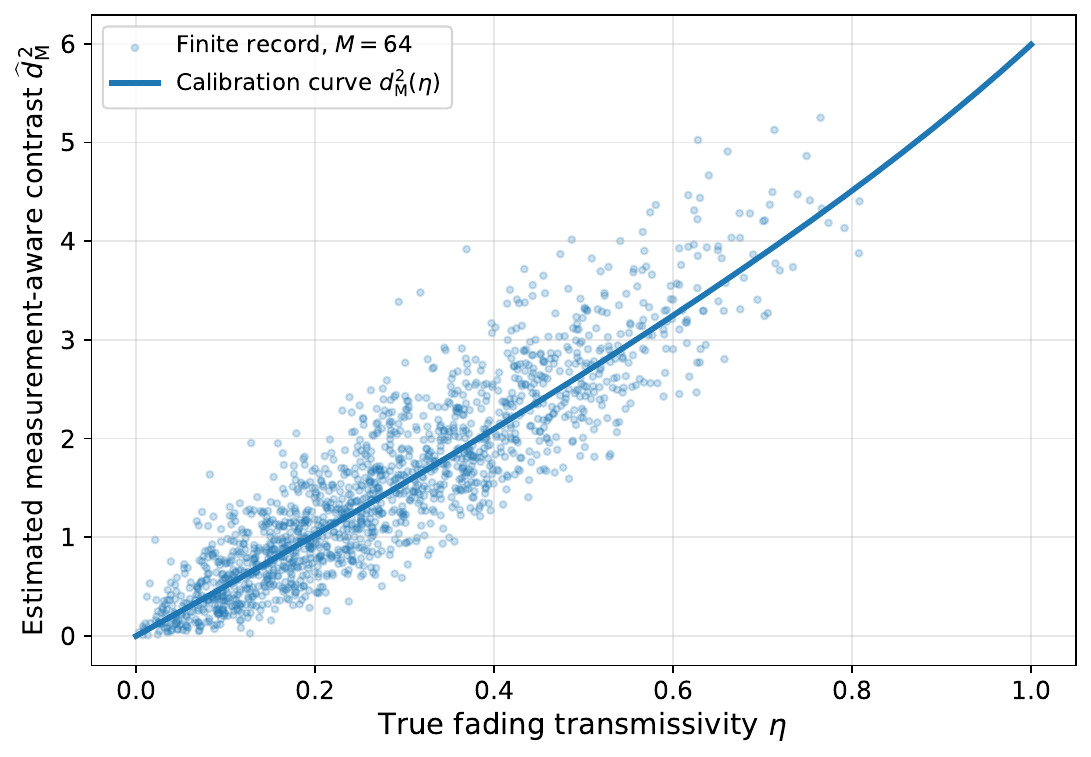}
\caption{\small Finite-sample calibration in the block-fading model. Each point is obtained from one $M=64$ sample homodyne record with an unknown $\eta$ drawn from Eq.~(\ref{eq:fading_law}); the solid line is the theoretical $d_{\mathcal M}^2(\eta)$ calibration curve in Eq.~(\ref{eq:dM_calibration}). The vertical scatter is therefore sampling uncertainty in the empirical measurement-aware contrast $\widehat d_{\mathcal M}^{\,2}$ rather than a distribution of exact distance values evaluated at known channel parameters.}
\label{fig:fading_calibration}
\end{figure}

To benchmark the transport estimator, we use two references. The \emph{mean-only estimator} ignores the sample variance and uses
\begin{equation}
\widehat\eta_{\mathrm{mean}}
:=\mathrm{clip}_{[0,1]}\!\left[\left(\frac{\widehat m}{A}\right)^2\right].
\label{eq:mean_estimator}
\end{equation}
The \emph{Gaussian maximum-likelihood estimator} uses the full likelihood of the homodyne record under Eqs.~(\ref{eq:fading_variance})--(\ref{eq:fading_samples}):
\begin{equation}
\widehat\eta_{\mathrm{MLE}}
:=\argmax_{\eta\in[0,1]}
\left\{-\frac{M}{2}\ln s_\eta
-\frac{1}{2s_\eta}\sum_{k=1}^{M}(X_k-m_\eta)^2\right\}.
\label{eq:mle_estimator}
\end{equation}
We quantify accuracy by the root-mean-square error
\begin{equation}
\mathrm{RMSE}(\widehat\eta):=
\sqrt{\mathbb E[(\widehat\eta-\eta)^2]},
\label{eq:rmse_definition}
\end{equation}
where the expectation is over both the block-fading law and conditional measurement samples; the plotted values are Monte Carlo estimates from $12000$ blocks for each $M$.

\begin{figure}[htbp]
\centering
\includegraphics[width=0.88\linewidth]{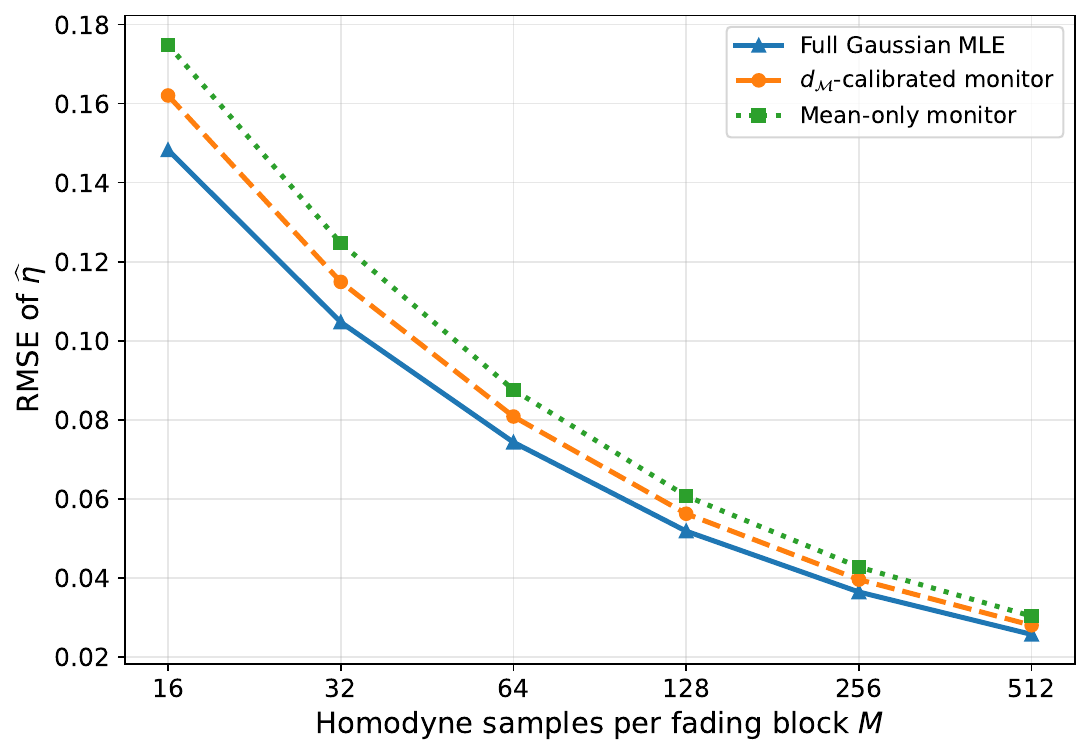}
\caption{\small RMSE for estimating the block transmissivity with the same physical parameters as Fig.~\ref{fig:fading_calibration}. The Gaussian MLE uses the full model likelihood, the $d_{\mathcal M}$-calibrated monitor uses the scalar statistic $\widehat d_{\mathcal M}^{\,2}$ defined in Eq.~(\ref{eq:empirical_dM}), and the mean-only estimator discards the variance information.}
\label{fig:fading_rmse}
\end{figure}

For all simulated sample sizes the ordering is
\begin{equation}
\mathrm{RMSE}(\widehat\eta_{\mathrm{MLE}})
<
\mathrm{RMSE}(\widehat\eta_{\mathcal M})
<
\mathrm{RMSE}(\widehat\eta_{\mathrm{mean}}).
\end{equation}

At $M=64$, for example, the Monte Carlo RMSEs are approximately $0.0743$, $0.0808$, and $0.0876$, respectively. The result establishes a concrete but limited role for the transport statistic: it compresses the empirical mean and variance into one calibrated scalar estimate of channel transmissivity, improves on a mean-only proxy for this model, and remains less efficient than the full likelihood estimator. We make no claim that $\dM$ retains all likelihood information or defines an optimal estimator.

The low-transmissivity expansions in Appendix~\ref{app:fading_asymptotics} explain the shape of the calibration curve. With nonzero coherent displacement,
\begin{equation}
\Wiso^2=2N_{\mathrm{disp}}\eta+O(\eta^2),
\end{equation}
and the aligned projected and measurement-aware contrasts are first order, whereas the orthogonal anti-squeezed projection begins at second order. The fidelity exponent is also first order,
\begin{equation}
-\ln F=\frac{N_{\mathrm{disp}}}{2n}\eta+O(\eta^2),
\label{eq:fidelity_small_eta}
\end{equation}
so $\Wiso^2/[-\ln F]\to4n$. Hence the application does not rely on a larger raw numerical span of the transport quantity; its operational content is the explicitly defined finite-sample estimator and its benchmarked RMSE.

% =================================================
\section{Discussion: frequency-dependent extension and scope}
\label{sec:frequency}
% =================================================

The block-fading application introduces temporal fluctuations of a single effective channel mode but does not address broadband sensing. This distinction is important for squeezing: in a frequency-dependent sensor, the signal transfer function, noise spectrum, and useful squeezing angle can all vary across the detection band, as occurs for example in gravitational-wave interferometry \cite{Ganapathy2023}. The single-mode conclusions above therefore do not imply a frequency-independent squeezing advantage or disadvantage.

For a finite vector $Y$ of jointly Gaussian filtered detector outputs, with mean $\mu$ and covariance $\Sigma$, the natural Euclidean-ground-cost extension of the measurement-level transport distance is
\begin{multline}
W_2^2=\|\mu_0-\mu_1\|^2+
\Trc\left[\Sigma_0+\Sigma_1\right.\\
\left.-2(\Sigma_0^{1/2}\Sigma_1\Sigma_0^{1/2})^{1/2}\right].
\end{multline}
The detector transfer functions and frequency-dependent noise correlations then enter through $\mu$ and $\Sigma$. Conventional spectral sensitivity criteria instead normalize a signal response by the appropriate noise spectral density or covariance. A broadband optimization must therefore specify the filters, frequency-domain measurement model, and ground cost before a transport criterion can be compared with spectral SNR or Fisher information.

The same caution applies beyond Gaussian probes. The definition $\dM=W_2(P_{\rho_0}^{\mathcal M},P_{\rho_1}^{\mathcal M})$ remains mathematically meaningful whenever a measurement produces classical output laws on a space with a specified quadratic ground cost. What changes outside the present model is analytical tractability, not the definition itself. A complementary environment-aware construction for non-Gaussian oscillator states adapts a local transport cost to selected dissipative directions and quotients out Hamiltonian tangents \cite{CoatanhayDremeau2026Environment}. It addresses a different contextual question from the measurement-aware hierarchy considered here: the former asks how efficiently a non-Hamiltonian variation is represented by a chosen environmental frame, whereas the present construction asks which finite differences remain visible after a specified measurement chain. The explicit propositions and numerical conclusions of this article, however, are restricted to the single-mode Gaussian thermal-loss setting stated in Sec.~\ref{sec:model}.

% =================================================
\section{Conclusion}
\label{sec:conclusion}
% =================================================

We have constructed and compared three transport levels for a single-mode Gaussian thermal-loss sensing model. The isotropic distance $\Wiso$ is the quadratic Wasserstein distance between the classical Gaussian laws associated with the complete first- and second-moment data. The projected distance $\Wtheta$ is the Wasserstein distance between the probability laws of one selected homodyne quadrature. The measurement-aware distance $\dM$ is defined after the detector model, directly on the measured output distributions. The latter two are genuine metrics on their classical output laws and generally pseudometrics when pulled back to quantum states. For the Gaussian homodyne model, projection and common additive detector noise are contractive, so $\dM\le \Wtheta\le \Wiso$.

At fixed total input energy, the isotropic objective locally decreases when energy is moved from coherent displacement to squeezing. By contrast, a projected transport contrast can be dominated either by the displaced/squeezed axis or by the anti-squeezed variance deformation, depending on the measured quadrature. This transport preference is not a metrological squeezing gain. The explicit CFI/QFI comparison confirms that finite-distance and local-estimation objectives can choose different probe allocations.

The block-fading application supplies the measurement-aware construction with a concrete statistical use. When $\eta$ is unknown and constant only over a finite measurement block, $\dM$ can be estimated from the sample mean and variance and inverted through a calibration curve. In the example studied here, the resulting scalar estimator improves on a mean-only estimator while remaining less efficient than the full Gaussian MLE. This benchmark, rather than the raw numerical span of the transport quantity, is the operational statement supported by the simulations.

The present results are deliberately limited to a single effective Gaussian mode, a thermal-loss channel, homodyne detection, and additive Gaussian detector noise. Extending the framework to broadband or non-Gaussian settings requires an explicit output space and ground cost at the measurement level. Within its stated domain, the hierarchy provides a controlled way to distinguish three questions that are often conflated: finite deformation of Gaussian state data, finite contrast in a selected measurement record, and local sensitivity for parameter estimation.

\appendix

\section{Proof of the local non-optimality proposition}
\label{app:proof_nonoptimality}

In this appendix, we derive the local expansion of the isotropic transport metric near zero squeezing and prove proposition \ref{prop:local_nonoptimality}. We use the thermal-loss model of Sec.~\ref{sec:model}, the isotropic distance of Sec.~\ref{sec:isotropic}, and the energy parametrization analyzed in Sec.~\ref{sec:nonoptimality}.

We start from the explicit isotropic transport formula
\begin{multline}
\Wiso(\rho_{\refstate},\rho_{\out})^2
=
2\eta \Ndisp
+
\left(\sqrt{v_x^{\out}}-\sqrt{n}\right)^2 \\
+
\left(\sqrt{v_p^{\out}}-\sqrt{n}\right)^2,
\label{eq:appendix_Wiso_start}
\end{multline}
where the state parameters are linked by
\begin{equation}
n = \nth+\frac12, \qquad \Ndisp = (1-\beta)\Ntot, \qquad \sinh^2 r = \beta\Ntot,
\end{equation}
and the output variances are given by
\begin{equation}
v_x^{\out} = \frac{\eta}{2}e^{-2r} + (1-\eta)n,
\qquad
v_p^{\out} = \frac{\eta}{2}e^{2r} + (1-\eta)n.
\label{eq:appendix_vxp}
\end{equation}
Our goal is to determine the sign of the right derivative
\begin{equation}
\frac{\partial}{\partial \beta} \Wiso(\rho_{\refstate},\rho_{\out})^2 \Big|_{\beta=0^+}.
\end{equation}

To study the limit $\beta\to 0^+$, it is mathematically convenient to introduce the auxiliary variable
\begin{equation}
a := \sqrt{\beta\Ntot},
\label{eq:appendix_a_def}
\end{equation}
so that $a^2 = \sinh^2 r$. The expansion can then be performed in $a$, which is smooth at the origin. Using $e^{\pm 2r} = \left(\sqrt{1+a^2}\pm a\right)^2$, we obtain the asymptotic behaviors
\begin{align}
\frac12 e^{-2r} &= \frac12 - a + a^2 + O(a^3), \label{eq:appendix_half_exp_minus} \\
\frac12 e^{2r} &= \frac12 + a + a^2 + O(a^3). \label{eq:appendix_half_exp_plus}
\end{align}

Let $s_0$ denote the output variance in the strict absence of squeezing:
\begin{equation}
s_0 := \frac{\eta}{2}+(1-\eta)n.
\label{eq:appendix_s0}
\end{equation}
Since $\eta\in(0,1]$ and $n>0$, we strictly have $s_0>0$. Substituting (\ref{eq:appendix_half_exp_minus}) and  (\ref{eq:appendix_half_exp_plus}) into (\ref{eq:appendix_vxp}) yields
\begin{align}
v_x^{\out} &= s_0 - \eta a + \eta a^2 + O(a^3), \label{eq:appendix_vx_expanded} \\
v_p^{\out} &= s_0 + \eta a + \eta a^2 + O(a^3). \label{eq:appendix_vp_expanded}
\end{align}
Concurrently, the displacement contribution scales as
\begin{equation}
2\eta\Ndisp = 2\eta\Ntot(1-\beta) = 2\eta\Ntot - 2\eta a^2.
\label{eq:appendix_displacement_term}
\end{equation}

To expand the covariance contribution, we define the auxiliary function
\begin{equation}
f(s) := (\sqrt{s}-\sqrt{n})^2 = s + n - 2\sqrt{sn},
\label{eq:appendix_f_def}
\end{equation}
so that the variance term in (\ref{eq:appendix_Wiso_start}) rewrites as $f(v_x^{\out})+f(v_p^{\out})$. The derivatives of $f$ are
\begin{equation}
f'(s) = 1-\sqrt{\frac{n}{s}}, \qquad f''(s) = \frac12\sqrt{n}\,s^{-3/2}.
\end{equation}
Because $s_0>0$, $f$ is smooth in a neighborhood of $s_0$. Expanding to second order around $s_0$, we write
\begin{multline}
f(v_{x,p}^{\out}) = f(s_0) + f'(s_0)\bigl(v_{x,p}^{\out}-s_0\bigr) \\
+ \frac12 f''(s_0)\bigl(v_{x,p}^{\out}-s_0\bigr)^2 + O(a^3).
\end{multline}
Using (\ref{eq:appendix_vx_expanded}) and (\ref{eq:appendix_vp_expanded}), the linear terms in $a$ perfectly cancel upon summation, leaving
\begin{multline}
f(v_x^{\out})+f(v_p^{\out})
=
2f(s_0) + \left( 2\eta f'(s_0)\right. \\+ \left.\eta^2 f''(s_0) \right)a^2
+ O(a^3).
\label{eq:appendix_covariance_sum}
\end{multline}
Noting that $a^2 = \beta\Ntot$ and $O(a^3) = O(\beta^{3/2})$, this becomes
\begin{multline}
f(v_x^{\out})+f(v_p^{\out})
=
2f(s_0) + \Ntot \left( 2\eta f'(s_0)\right. \\+ \left.\eta^2 f''(s_0) \right)\beta + O(\beta^{3/2}).
\label{eq:appendix_covariance_beta}
\end{multline}

Combining the displacement term (\ref{eq:appendix_displacement_term}) and the covariance sum (\ref{eq:appendix_covariance_beta}) into (\ref{eq:appendix_Wiso_start}), the full metric expands as
\begin{multline}
\Wiso(\rho_{\refstate},\rho_{\out})^2
=
\Wiso(\beta=0)^2
\\
+\Ntot \left[ -2\eta + 2\eta f'(s_0) + \eta^2 f''(s_0) \right]\beta + O(\beta^{3/2}).
\label{eq:appendix_full_expansion}
\end{multline}
The right derivative at $\beta=0$ therefore exists and equates to the coefficient of the linear term:
\begin{multline}
\frac{\partial}{\partial\beta} \Wiso(\rho_{\refstate},\rho_{\out})^2 \Big|_{\beta=0^+}
=
\Ntot \left[ -2\eta \right. \\+ \left.2\eta f'(s_0) + \eta^2 f''(s_0) \right].
\label{eq:appendix_beta_derivative_step1}
\end{multline}
Substituting the explicit forms of $f'(s_0)$ and $f''(s_0)$ yields
\begin{equation}
\frac{\partial}{\partial\beta} \Wiso(\rho_{\refstate},\rho_{\out})^2 \Big|_{\beta=0^+}
=
\Ntot \left[ -2\eta\sqrt{\frac{n}{s_0}} + \frac{\eta^2}{2}\sqrt{n}\,s_0^{-3/2} \right].
\label{eq:appendix_beta_derivative_step2}
\end{equation}
Factoring out the strictly positive term $\Ntot\eta\sqrt{n}\,s_0^{-3/2}$, we obtain
\begin{equation}
\frac{\partial}{\partial\beta} \Wiso(\rho_{\refstate},\rho_{\out})^2 \Big|_{\beta=0^+}
=
\Ntot\eta\sqrt{n}\,s_0^{-3/2} \left( \frac{\eta}{2} - 2s_0 \right).
\label{eq:appendix_beta_derivative_factorized}
\end{equation}
Finally, expanding the bracket via the definition of $s_0$ gives
\begin{equation}
\frac{\eta}{2} - 2s_0 = \frac{\eta}{2} - 2\left( \frac{\eta}{2} + (1-\eta)n \right) = -\frac{\eta}{2} - 2(1-\eta)n.
\label{eq:appendix_sign_bracket}
\end{equation}
Because the physical parameters satisfy $\eta\in(0,1]$ and $n>0$, this bracket is strictly negative. Consequently,
\begin{equation}
\frac{\partial}{\partial\beta} \Wiso(\rho_{\refstate},\rho_{\out})^2 \Big|_{\beta=0^+} < 0,
\end{equation}
which concludes the proof of proposition \ref{prop:local_nonoptimality}.

The physical interpretation of this negative sign is unambiguous: at first order, the geometric cost of anti-squeezing strictly dominates the benefits of variance reduction along the squeezed axis. This confirms that while the isotropic metric evaluates global phase-space deformation, it inherently penalizes the directional nature of squeezed resources.

% =================================================
\section{Low-transmissivity asymptotics}
\label{app:fading_asymptotics}
% =================================================

In this appendix, we derive the small-transmissivity expansions used in the strong-loss comparison and block-fading application of Sec.~\ref{sec:fading}. The primary objective is to demonstrate that, in the strong-loss regime, the isotropic and projected transport metrics do not scale identically with the channel transmissivity. This asymptotic behavior explains the scaling of the transport calibration used in the finite-sample fading analysis. We additionally compare it with the leading fidelity exponent in order to separate asymptotic order from arbitrary numerical scale.

\subsection{Output variances at small transmissivity}

Recall that the output variances of the thermal-loss channel are given by
\begin{equation}
v_x^{\out}(\eta)=\frac{\eta}{2}e^{-2r}+(1-\eta)n,
\qquad
v_p^{\out}(\eta)=\frac{\eta}{2}e^{2r}+(1-\eta)n,
\end{equation}
where $n=\nth+\frac12$. To isolate the linear dependence on $\eta$, it is convenient to introduce the purely squeezed variance components
\begin{equation}
a_x:=\frac12 e^{-2r},
\qquad
a_p:=\frac12 e^{2r},
\end{equation}
so that the output variances rewrite exactly as
\begin{equation}
v_x^{\out}(\eta)=n+\eta(a_x-n),
\qquad
v_p^{\out}(\eta)=n+\eta(a_p-n).
\label{eq:appendix_eta_variances}
\end{equation}

In the strong-loss limit ($\eta \to 0$), the variance deviation $\delta \propto \eta$ is small. We rely on the standard Taylor expansion
\begin{equation}
\sqrt{n+\delta}
=
\sqrt n+\frac{\delta}{2\sqrt n}-\frac{\delta^2}{8n^{3/2}}+O(\delta^3),
\qquad
\delta\to 0,
\end{equation}
which immediately implies that the squared difference scales identically at leading order:
\begin{equation}
(\sqrt{n+\delta}-\sqrt n)^2
=
\frac{\delta^2}{4n}+O(\delta^3).
\label{eq:appendix_sqrt_square_expansion}
\end{equation}

\subsection{Isotropic metric}

Applying this expansion to the isotropic transport metric, we evaluate the covariance penalty:
\begin{multline}
\Wiso\bigl(\rho_{\refstate},\rho_{\out}(\eta)\bigr)^2
=
2\eta\Ndisp
+
\left(\sqrt{v_x^{\out}(\eta)}-\sqrt n\right)^2
\\+
\left(\sqrt{v_p^{\out}(\eta)}-\sqrt n\right)^2.
\end{multline}
Using (\ref{eq:appendix_eta_variances}) and (\ref{eq:appendix_sqrt_square_expansion}), the variance contributions expand as
\begin{equation}
\left(\sqrt{v_x^{\out}(\eta)}-\sqrt n\right)^2
=
\frac{\eta^2}{4n}(a_x-n)^2+O(\eta^3),
\end{equation}
and identically for the $p$-quadrature. Summing these terms yields the full expansion
\begin{multline}
\Wiso\bigl(\rho_{\refstate},\rho_{\out}(\eta)\bigr)^2
=
2\eta\Ndisp
\\+
\frac{\eta^2}{4n}\Bigl[(a_x-n)^2+(a_p-n)^2\Bigr]
+
O(\eta^3).
\label{eq:appendix_Wiso_eta_expansion}
\end{multline}
In particular, the isotropic metric retains a linear leading-order term:
\begin{equation}
\Wiso\bigl(\rho_{\refstate},\rho_{\out}(\eta)\bigr)^2
=
2\eta\Ndisp+O(\eta^2).
\label{eq:appendix_Wiso_first_order}
\end{equation}

\subsection{Projected metric at fixed angle}

For a fixed observation angle $\theta$, the projected metric reads
\begin{equation}
\Wtheta\bigl(\rho_{\refstate},\rho_{\out}(\eta)\bigr)^2
=
2\eta\Ndisp\cos^2\theta
+
\left(
\sqrt{v_\theta^{\out}(\eta)}-\sqrt n
\right)^2,
\end{equation}
where the projected output variance is $v_\theta^{\out}(\eta) = v_x^{\out}(\eta)\cos^2\theta+v_p^{\out}(\eta)\sin^2\theta$.
Defining the effective angle-dependent squeezed variance
\begin{equation}
a_\theta:=a_x\cos^2\theta+a_p\sin^2\theta,
\end{equation}
we similarly obtain $v_\theta^{\out}(\eta)=n+\eta(a_\theta-n)$. The projected metric therefore expands as
\begin{equation}
\Wtheta\bigl(\rho_{\refstate},\rho_{\out}(\eta)\bigr)^2
=
2\eta\Ndisp\cos^2\theta
+
\frac{\eta^2}{4n}(a_\theta-n)^2
+
O(\eta^3).
\label{eq:appendix_Wtheta_eta_expansion}
\end{equation}

Evaluating this expression at the two candidate boundary axes reveals a strict structural difference:
\begin{multline}
\Wtheta\bigl(\rho_{\refstate},\rho_{\out}(\eta)\bigr)^2\Big|_{\theta=0}
=
2\eta\Ndisp
\\+
\frac{\eta^2}{4n}(a_x-n)^2
+
O(\eta^3),
\label{eq:appendix_Wtheta0_eta}
\end{multline}
\begin{multline}
\Wtheta\bigl(\rho_{\refstate},\rho_{\out}(\eta)\bigr)^2\Big|_{\theta=\pi/2}
=
\frac{\eta^2}{4n}(a_p-n)^2
+
O(\eta^3).
\label{eq:appendix_Wtheta90_eta}
\end{multline}
(\ref{eq:appendix_Wtheta0_eta}) confirms that the quadrature aligned with the coherent displacement preserves a first-order sensitivity to the channel transmissivity. Conversely, (\ref{eq:appendix_Wtheta90_eta}) demonstrates that the orthogonal anti-squeezed quadrature provides a geometric contrast only at second order. This mathematically explains the severe compression observed numerically for the orthogonal projection under strong loss.

\subsection{Measurement-aware metric with additive detector noise}

For the noisy quadrature measurement chain introduced in Sec.~\ref{sec:measurement}, the measurement-aware transport metric is given by
\begin{multline}
\dM\bigl(\rho_{\refstate},\rho_{\out}(\eta)\bigr)^2
=
2\eta\Ndisp\cos^2\theta
\\+
\left(
\sqrt{v_\theta^{\out}(\eta)+\sigma_d^2}
-
\sqrt{n+\sigma_d^2}
\right)^2.
\end{multline}
Setting the noise-regularized reference variance as $N:=n+\sigma_d^2$, the identical Taylor expansion yields
\begin{equation}
\dM\bigl(\rho_{\refstate},\rho_{\out}(\eta)\bigr)^2
=
2\eta\Ndisp\cos^2\theta
+
\frac{\eta^2}{4N}(a_\theta-n)^2
+
O(\eta^3).
\label{eq:appendix_Wmeas_eta_expansion}
\end{equation}
Consequently, additive detector noise merely rescales the quadratic coefficient ($n \mapsto N$), but leaves the fundamental first-order versus second-order geometric hierarchy completely intact.

\subsection{Fidelity exponent at small transmissivity}

For comparison with the overlap benchmark, consider the exact single-mode Gaussian fidelity between the thermal reference and the received state. The output displacement satisfies
\begin{equation}
\Delta\nu=(\sqrt{2N_{\mathrm{disp}}\eta},0)^T,
\end{equation}
while the covariance differs from the reference only by $O(\eta)$. In the Gaussian fidelity formula \cite{Banchi2015}, the displacement factor therefore gives the only first-order contribution to $-\ln F$:
\begin{align}
-\ln F
&=\frac12\Delta\nu^T(V_{\mathrm{ref}}+V_{\mathrm{out}})^{-1}\Delta\nu+O(\eta^2)\nonumber\\
&=\frac{N_{\mathrm{disp}}}{2n}\eta+O(\eta^2),
\label{eq:appendix_fidelity_eta}
\end{align}
where $V_{\mathrm{ref}}=nI_2$ and $V_{\mathrm{out}}=nI_2+O(\eta)$. Combining Eq.~(\ref{eq:appendix_fidelity_eta}) with Eq.~(\ref{eq:appendix_Wiso_first_order}) yields
\begin{equation}
\frac{\Wiso^2}{-\ln F}\longrightarrow 4n
\qquad (\eta\to0,\ N_{\mathrm{disp}}>0).
\end{equation}
Thus the two quantities have the same leading order when displacement is present; their different raw amplitudes are not, by themselves, a performance advantage.

\subsection{Consequences for block-fading statistics}

To rigorously connect these expansions to the block-fading application, we model the strong-loss statistical regime by introducing the random variable
\begin{equation}
\eta=\varepsilon Z,
\qquad
0\le Z\le 1,
\qquad
\varepsilon\ll 1,
\end{equation}
where $Z$ follows a fixed bounded distribution. From (\ref{eq:appendix_Wiso_first_order}), the block-fading isotropic contrast scales as
\begin{equation}
\Wiso\bigl(\rho_{\refstate},\rho_{\out}(\varepsilon Z)\bigr)^2
=
2\varepsilon\Ndisp Z
+
O(\varepsilon^2).
\end{equation}
Taking the expectation and variance over the random variable $Z$, we obtain
\begin{align}
\mathbb E\!\left[\Wiso\bigl(\rho_{\refstate},\rho_{\out}(\eta)\bigr)^2\right]
&=
2\varepsilon\Ndisp\,\mathbb E[Z]
+
O(\varepsilon^2),
\label{eq:appendix_mean_Wiso}
\\
\mathrm{Var}\!\left(\Wiso\bigl(\rho_{\refstate},\rho_{\out}(\eta)\bigr)^2\right)
&=
4\varepsilon^2\Ndisp^2\,\mathrm{Var}(Z)
+
O(\varepsilon^3).
\label{eq:appendix_var_Wiso}
\end{align}

Applying the same statistical treatment to the projected boundaries using (\ref{eq:appendix_Wtheta0_eta}) and (\ref{eq:appendix_Wtheta90_eta}) yields
\begin{equation}
\mathbb E\!\left[\Wtheta\bigl(\rho_{\refstate},\rho_{\out}(\eta)\bigr)^2\Big|_{\theta=0}\right]
=
2\varepsilon\Ndisp\,\mathbb E[Z]
+
O(\varepsilon^2),
\label{eq:appendix_mean_Wtheta0}
\end{equation}
whereas the orthogonal projection falls off rapidly:
\begin{multline}
\mathbb E\!\left[\Wtheta\bigl(\rho_{\refstate},\rho_{\out}(\eta)\bigr)^2\Big|_{\theta=\pi/2}\right]
=
\frac{\varepsilon^2}{4n}(a_p-n)^2\,\mathbb E[Z^2]
\\+
O(\varepsilon^3),
\label{eq:appendix_mean_Wtheta90}
\end{multline}
\begin{multline}
\mathrm{Var}\!\left(\Wtheta\bigl(\rho_{\refstate},\rho_{\out}(\eta)\bigr)^2\Big|_{\theta=\pi/2}\right)
=
O(\varepsilon^4).
\label{eq:appendix_var_Wtheta90}
\end{multline}

These relations show that, in the strong-loss scaling $\eta=\varepsilon Z$, the isotropic contrast and the displaced-axis projection inherit channel fluctuations at first order in $\varepsilon$, whereas the orthogonal anti-squeezed-axis projection begins at second order. This establishes the directional scaling used to interpret the finite-sample calibration in Sec.~\ref{sec:fading}; it does not by itself imply a metrological advantage.

% -------------------------------------------------
% BIBLIOGRAPHY
% -------------------------------------------------
%\bibliographystyle{abbrv}
\bibliography{references}

\end{document}